\documentclass[smallextended]{svjour3}       
\usepackage{amsmath,amssymb,algorithm}
\usepackage[noend]{algorithmic}
\usepackage[usenames,dvipsnames]{xcolor}
\usepackage[colorlinks=true,pdfpagemode=UseNone,urlcolor=blue,linkcolor=blue,citecolor=BrickRed,pdfstartview=FitH]{hyperref}

\newcommand{\bbN}{\mathbb{N}}
\newcommand{\bbR}{\mathbb{R}}
\newcommand{\bbZ}{\mathbb{Z}}
\newcommand{\bfzero}{\mathbf{0}}
\newcommand{\bfone}{\mathbf{1}}

\newcommand{\biDelta}{\boldsymbol{\Delta}}
\newcommand{\bic}{\boldsymbol{c}}
\newcommand{\bid}{\boldsymbol{d}}

\newcommand{\bip}{\boldsymbol{p}}

\newcommand{\biu}{\boldsymbol{u}}
\newcommand{\bix}{\boldsymbol{x}}
\newcommand{\biy}{\boldsymbol{y}}
\newcommand{\biz}{\boldsymbol{z}}
\newcommand{\biw}{\boldsymbol{w}}

\newcommand{\bichi}{\boldsymbol{\chi}}

\newcommand{\caD}{\mathcal{D}}

\newcommand{\E}{\mathop{\mathbf{E}}}
\newcommand{\bit}{\{0,1\}}
\newcommand{\OPT}{\mathrm{OPT}}
\newcommand{\supp}{\mathrm{supp}}
\newcommand{\argmax}{\mathrm{argmax}}
\newcommand{\wmin}{w_{\min}}
\newcommand{\compknap}{O\left(\frac{n^3}{\epsilon^3}\log^3\tau \left[ \log^3\|\bic\|_\infty + \frac{n}{\epsilon}\log\|\bic\|_\infty \log\frac{1}{\epsilon\wmin} \right]\right)}

\title{Maximizing Monotone Submodular Functions over the Integer Lattice\thanks{T.S. is supported by JSPS Grant-in-Aid for JSPS Fellows. Y.Y. is supported by JSPS Grant-in-Aid for Young Scientists (B) (No.~26730009), MEXT Grant-in-Aid for Scientific Research on Innovative Areas (No.~24106003), and JST, ERATO, Kawarabayashi Large Graph Project.}
}
\titlerunning{Maximizing Monotone Submodular Functions over $\bbZ_+^E$}
\author{Tasuku Soma         \and
        Yuichi Yoshida 
}
\institute{T. Soma \at
              Graduate School of Information Science and Technology, the University of Tokyo\\
             \email{tasuku\_soma@mist.i.u-tokyo.ac.jp}
           \and
           Y. Yoshida \at
    National Institute of Informatics and Preferred Infrastructure, Inc. \\
     \email{yyoshida@nii.ac.jp}
}
\date{Received: date / Accepted: date}

\begin{document}


\maketitle
\begin{abstract}
  The problem of maximizing non-negative monotone submodular functions under a certain constraint has been intensively studied in the last decade.
  In this paper, we address the problem for functions defined over the integer lattice.

  Suppose that a non-negative monotone submodular function $f:\bbZ_+^n \to \bbR_+$ is given via an evaluation oracle.
  Assume further that $f$ satisfies the diminishing return property, which is not an immediate consequence of submodularity when the domain is the integer lattice.
  Given this, we design polynomial-time $(1-1/e-\epsilon)$-approximation algorithms for a cardinality constraint, a polymatroid constraint, and a knapsack constraint.
  For a cardinality constraint, we also provide a $(1-1/e-\epsilon)$-approximation algorithm with slightly worse time complexity that does not rely on the diminishing return property.

\end{abstract}

\section{Introduction}\label{sec:intro}

Submodular functions have been intensively studied in various areas of operations research and computer science, as submodularity naturally arises in many problems in these fields~\cite{Fujishige2005,Iwata2007survey,Krause2014survey}.
In the last decade, the \emph{maximization} of submodular functions in particular has attracted interest.
For example, one can find novel applications of submodular function maximization in the dissemination of influence through social networks~\cite{Kempe2003}, text summarization~\cite{Lin2010,Lin2011}, and optimal budget allocation for advertisements~\cite{Alon2012}.

Most past works in the area have considered submodular functions defined over a set---submodular functions which take a subset of a ground set as the input and return a real value.
However, in many practical scenarios, it is more natural to consider submodular functions over a multiset or, equivalently, submodular functions over the integer lattice $\bbZ^E$ for some finite set $E$.
We say that a function $f:\bbZ^E \to \bbR$ is \emph{(lattice) submodular} if $f(\bix) + f(\biy) \geq f(\bix\vee\biy) + f(\bix\wedge\biy)$ for all $\bix, \biy \in \bbZ^E$, where $\bix\vee \biy$ and $\bix\wedge \biy$ denote the coordinate-wise maximum and minimum, respectively.
Such a generalized form of submodularity arises in maximizing the spread of influence with partial incentives~\cite{Demaine:2014dw},
optimal budget allocation, sensor placement, and text summarization~\cite{Soma:2014tp}.

When designing algorithms for maximizing submodular functions,
the \emph{diminishing return property} often plays a crucial role.
A set function $f:2^E \to \bbR$ is said to satisfy the diminishing return property if $f(X+e) - f(X) \geq f(Y+e) - f(Y)$ for all $X \subseteq Y \subseteq E$ and $e \notin Y$.
For example, the simple greedy algorithm for cardinality constraints proposed by Nemhauser~\emph{et~al.}~\cite{Nemhauser1978} works because of this property.
For set functions,
it is well-known that submodularity is equivalent to the diminishing return property.
For functions over the integer lattice,
however,
lattice submodularity only implies a weaker variant of the inequality.
This causes difficulty in designing approximation algorithms; even for a single cardinality constraint, we need a more complicated approach such as partial enumeration~\cite{Alon2012,Soma:2014tp}.

Fortunately, objective functions appearing in practical applications admit the diminishing return property in the following sense.
We say that a function $f:\bbZ^E \to \bbR$ is \emph{diminishing return submodular} (\emph{DR-submodular}) if $f(\bix+\bichi_e)-f(\bix) \geq f(\biy+\bichi_e)-f(\biy)$ for arbitrary $\bix\leq\biy$ and $e\in E$, where $\bichi_e$ is the $i$-th unit vector.
Any DR-submodular function is lattice submodular; i.e., DR-submodularity is stronger than lattice submodularity.\footnote{Note that $f$ is DR-submodular if and only if it is lattice submodular and satisfies the \emph{coordinate-wise concave condition}: $f(\bix+\bichi_e) - f(\bix) \geq f(\bix+2\bichi_e) - f(\bix+\bichi_e)$ for any $\bix$ and $e \in E$ (see~\cite[Lemma 2.3]{sfcover2015}).}
The problem of maximizing DR-submodular functions over $\bbZ^E$ naturally appears in the submodular welfare problem~\cite{Kapralov2012,Shioura2009} and the budget allocation problem with decreasing influence probabilities~\cite{Soma:2014tp}.
Nevertheless, only a few studies have considered this problem.
In fact,
it was not known whether we $(1-1/e)$-approximation can be obtained in polynomial time under a single cardinality constraint.


\subsection{Main Results}
In this paper, we develop polynomial-time approximation algorithms for maximizing monotone DR-submodular functions under cardinality constraints, polymatroid constraints, and knapsack constraints.
Let $f:\bbZ^E \to \bbR$ be a non-negative monotone DR-submodular function unless explicitly stated otherwise.
Then given any small constant $\epsilon>0$,
our algorithms find ($1-1/e-\epsilon$)-approximate solutions under these constraints.
The details are described below.


\begin{description}
    \item[Cardinality Constraint:]
    The objective is to maximize $f(\bix)$ subject to $\bfzero\leq\bix\leq\bic$ and $\bix(E)\leq r$, where $\bic\in\bbZ_+^E$, $r \in\bbZ_+$, and $\bix(E) = \sum_{e\in E}\bix(e)$.
    We design a deterministic approximation algorithm with $O(\frac{n}{\epsilon}\log\|\bic\|_\infty \log\frac{r}{\epsilon})$ running time, which is the first polynomial time algorithm for this problem.

\item[Cardinality Constraint (lattice submodular case):]
    For cardinality constraints, we also show a $(1-1/e-\epsilon)$-approximation algorithm for a monotone \emph{lattice submodular} function $f$.
    This algorithm runs in $O(\frac{n}{\epsilon^2}\log \|\bic\|_\infty \log \frac{r}{\epsilon} \log \tau)$ time, where $\tau$ is the ratio of the maximum value of $f$ to the minimum positive increase in the value of $f$.

\item[Polymatroid Constraint:]
    The objective is to maximize $f(\bix)$ subject to $\bix \in P \cap \bbZ^E_+$, where $P$ is a polymatroid given via an \emph{independence oracle}.
    Our algorithm runs in $\widetilde{O}(\frac{n^3}{\epsilon^5}\log^2 r + n^8)$ time, where $r$ is the maximum value of $\bix(E)$ for $\bix\in P$.
    This is the first polynomial time ($1-1/e-\epsilon$)-approximation algorithm for this problem.


\item[Knapsack Constraint:] The objective is to maximize $f(\bix)$ subject to $\bfzero \leq \bix \leq \bic$ and a single knapsack constraint $\biw^\top \bix \leq 1$, where $\biw\in (0,1]^E$.
    We devise an approximation algorithm with $\compknap$ running time, where $\tau$ is the ratio of the maximum value of $f$ to the minimum positive increase in the value of $f$, and $\wmin$ is the minimum entry of $\biw$. This is the first polynomial time algorithm for this problem.
\end{description}


\subsection{Technical Contribution}
In order to devise polynomial-time algorithms instead of pseudo-polynomial time algorithms, we need to combine several techniques carefully.
Our algorithms adapt the ``decreasing-threshold greedy'' framework recently introduced by Badanidiyuru and Vondr\'{a}k~\cite{Badanidiyuru:2013jc}, and work in the following way.
We maintain a feasible solution $\bix \in \bbR^E$ and a threshold $\theta \in \bbR$ during the algorithm.
Starting from $\bix=\mathbf{0}$, we greedily increase each component of $\bix$ if the average gain in the increase is above the threshold $\theta$, with consideration of constraints.
Slightly decreasing the threshold $\theta$, we repeat this greedy process until $\theta$ becomes sufficiently small.
We combine this framework with pseudo-polynomial time greedy algorithms to design a polynomial time algorithm for cardinality constraints.
We also need to incorporate the \emph{partial enumeration} technique~\cite{Soma:2014tp,Sviridenko2004} to obtain a polynomial time algorithm for knapsack constraints.
In order to develop a polynomial time algorithm for polymatroid constraints, we follow the \emph{continuous greedy} approach~\cite{Chekuri2010};
instead of the discrete problem, we consider the problem of maximizing a continuous extension of the original objective function.
After the greedy phase, we round the current fractional solution to an integral solution if needed.

As described above, our algorithms share some ideas with the algorithms of~\cite{Badanidiyuru:2013jc,Chekuri2010,Soma:2014tp,Sviridenko2004}.
However, we attain several improvements and introduce new ideas, mainly due to the essential difference between set functions and functions over the integer lattice.

\begin{description}
    \item[Binary Search in the Greedy Phase:]
        In most previous algorithms, the greedy step works as follows:
        find the direction of the maximum marginal gain and move the current solution along this direction with a \emph{unit} step size.
        However, it turns out that a naive adaptation of this greedy strategy only yields a pseudo-polynomial time algorithm.
        To circumvent this issue, we perform a binary search to determine the step size in the greedy phase.
        Combined with the decreasing threshold framework, this technique significantly reduces the time complexity.

    \item[New Continuous Extension:]
        To execute the continuous greedy algorithm, we need a continuous extension of functions over the integer lattice.
        Note that the \emph{multilinear extension}~\cite{Calinescu2011} cannot be directly used because the domain of the multilinear extension is only the hypercube $[0,1]^E$.
        In this paper, we propose a new continuous extension of a function over the integer lattice for polymatroid constraints.
        This continuous extension has similar properties to the multilinear extension when $f$ is DR-submodular,
        and is carefully designed so that we can round fractional solutions without violating polymatroid constraints.
        To the best of our knowledge, this continuous extension in $\bbR_+^E$ has not been proposed in the literature so far.

    \item[Rounding without violating polymatroid constraints:] 
        Rounding fractional solutions in $\bbR_+^E$ without violating polymatroid constraints is non-trivial.
        We show that the rounding can be reduced to rounding in a matroid polytope; therefore we can use existing rounding methods for a matroid polytope.
\end{description}

\paragraph{Modification to the conference version}
An extended abstract of this paper appeared in \cite{Soma2016:sfm}. 
Unfortunately, the algorithm for a knapsack constraint presented there is quite complicated and has a technical flaw: the correct time complexity is not as stated.
In this paper, we provide another much simpler algorithm for a knapsack constraint. 
A main difference is that the algorithm in this paper use partial enumeration, whereas the algorithm in \cite{Soma2016:sfm} used continuous greedy.


\subsection{Related Work}
Studies on maximizing monotone submodular functions were pioneered by Neumhauser, Wolsey, and Fisher~\cite{Nemhauser1978}.
They showed that a \emph{greedy} algorithm achieves a ($1-1/e$)-approximation for maximizing a monotone and submodular set function under a cardinality constraint, and a $1/2$-approximation under a matroid constraint.
Their algorithm provided a prototype for subsequent work.
For knapsack constraints, Sviridenko~\cite{Sviridenko2004} devised the first ($1-1/e$)-approximation algorithm with $O(n^5)$ running time.
Whereas these algorithms are combinatorial and deterministic,
the best known algorithms for matroid constraints are based on a continuous and randomized method.
The first ($1-1/e$)-approximation algorithm for a matroid constraint was provided by \cite{Calinescu2011}, and employed the continuous greedy approach:
first solve a continuous relaxation problem and obtain a fractional approximate solution; then round it to an integral feasible solution.
In their framework, the multilinear extension of a submodular set function was used as the objective function in the relaxation problem.
They also provided the \emph{pipage rounding} to obtain an integral feasible solution.
Chekuri, Vondr\'{a}k, and Zenklusen~\cite{Chekuri2010} designed a simple rounding method---\emph{swap rounding}---based on the exchange property of matroid base families.
Badanidiyuru and Vondr\'{a}k~\cite{Badanidiyuru:2013jc} recently devised $(1-1/e-\epsilon)$-approximation algorithms for any fixed constraint $\epsilon > 0$, with significantly lower time complexity for various constraints.
For the inapproximability side, Nemhauser~et~al.~\cite{Nemhauser1978} proved that no algorithm making polynomially many queries to a value oracle of $f$ can achieve an approximation ratio better than $1-1/e$ under any of the constraints mentioned so far.
Furthermore, Feige~\cite{Feige1998} showed that, even if $f$ is given explicitly, ($1-1/e$)-approximation is the best possible unless $\mathrm{P}=\mathrm{NP}$.

Generalized forms of submodularity have been studied in various contexts.
Fujishige~\cite{Fujishige2005} discussed submodular functions over a distributive lattice and its related polyhedra.
In the theory of discrete convex analysis by Murota~\cite{Murota2003}, a subclass of submodular functions over the integer lattice was considered.
The maximization problem has also been studied for variants of submodular functions.
Shioura~\cite{Shioura2009} investigated the maximization of discrete convex functions.
Soma~\emph{et~al.}~\cite{Soma:2014tp} provided a ($1-1/e$)-approximation algorithm for maximizing a monotone lattice submodular function under a knapsack constraint.
However, its running time is pseudo-polynomial.
Although this paper focuses on monotone submodular functions, there is a large body of work on maximization of \emph{non-monotone} submodular functions~\cite{Feige2011,Buchbinder2012,Buchbinder2014,Buchbinder2015}.
Gottschalk and Peis~\cite{Gottschalk2015} provided a $1/3$-approximation algorithm for maximizing a lattice submodular function over a (bounded) integer lattice.
Recently, \emph{bisubmodular} functions and \emph{$k$-submodular} functions, other generalizations of submodular functions, have been studied as well, and approximation algorithms for maximizing these functions can be found in~\cite{Iwata2013,Singh2012,Ward2014}.

\subsection{Organization of This Paper}
The rest of this paper is organized as follows.
In Section~\ref{sec:pre}, we provide our notations and basic facts on submodular functions and polymatroids.
Section~\ref{sec:cardinality} describes our algorithm for cardinality constraints.
In Section~\ref{sec:polymatroid}, we provide the continuous extension for polymatroid constraints and our approximate algorithm.
We present our algorithm for knapsack constraints in Section~\ref{sec:knapsack}.

\section{Preliminaries}\label{sec:pre}

\paragraph{Notation}
We denote the sets of non-negative integers and non-negative reals by $\bbZ_+$ and $\bbR_+$, respectively.
We denote the set of positive integers by $\bbN$.
For a positive integer $k \in \bbN$, $[k]$ denotes the set $\{1,\ldots,k\}$.
Throughout this paper, $E$ denotes a ground set of size $n$.
We denote the $i$-th entry of a vector $\bix \in \bbR^n$ by $\bix(i)$.
The $i$-th standard unit vector is denoted by $\bichi_i$.
The zero vector is denoted by $\bfzero$ and the all-one vector by $\bfone$.
We denote the characteristic vector of $X\subseteq E$ by $\bichi_X$.
For $f:\bbR^E \to \bbR$ and $\bix, \biy \in \bbR^E$, we define $f(\bix \mid \biy) := f(\bix + \biy) - f(\biy)$.
For $\bix \in \bbR^E$ and $X\subseteq E$, we denote $\bix(X):=\sum_{i\in X}\bix(i)$.
For a vector $\bix \in \bbR^E$, $\supp^+(\bix)$ denotes the set $\{e \in E \mid \bix(e) > 0\}$.
For $\bix \in \bbZ_+^E$, $\{\bix\}$ denotes the multiset where the element $e$ appears $\bix(e)$ times.
For arbitrary two multisets $\{\bix\}$ and $\{\biy\}$, we define $\{\bix \}\setminus\{\biy \} := \{(\bix-\biy)\vee\bfzero\}$.
For a multiset $\{\bix\}$, we define $|\{\bix\} | := \bix(E)$.
We often denote an optimal solution and the optimal value by $\bix^*$ and $\OPT$, respectively.

For an error parameter $\epsilon > 0$,
we always assume that $\frac{1}{\epsilon}$ is an integer (otherwise we can slightly reduce it without changing the asymptotic time complexity and the approximation ratio).

We frequently use the following form of Chernoff's bound.
\begin{lemma}[Relative+Additive Chernoff's bound,~\cite{Badanidiyuru:2013jc}]\label{lem:chernoff}
  Let $X_1 , \ldots , X_m$ be independent random variables such that for each $i$, $X_i \in [0,1]$.
  Let $X = \frac{1}{m} \sum X_i$ and $\mu = \E[X]$.
  Then
  \begin{align*}
    \Pr[X > (1+\alpha)\mu+\beta] \leq e^{-\frac{m\alpha \beta}{3}}, \\
    \Pr[X < (1-\alpha)\mu-\beta] \leq e^{-\frac{m\alpha \beta}{2}}.
  \end{align*}
\end{lemma}

\subsection{Submodularity and the Diminishing Return Property}
We say that a function $f:\bbZ_+^E \to \bbR$ is \emph{lattice submodular} if it satisfies $f(\bix) + f (\biy) \geq f(\bix \vee \biy)+f(\bix \wedge \biy)$ for all $\bix, \biy \in \bbZ^E$,
where $\bix\vee \biy$ and $\bix\wedge \biy$ denote the coordinate-wise maximum and minimum, respectively,
i.e., $(\bix \vee \biy)(e) = \max\{\bix(e), \biy(e)\}$ and $(\bix \wedge \biy)(e) = \min\{\bix(e), \biy(e)\}$ for each $e\in E$.
A function $f:\bbZ_+^E \to \bbR$ is \emph{monotone} if $f (\bix) \leq f (\biy)$ for all $\bix$ and $\biy$ with $\bix \leq \biy$.
We say that $f:\bbZ^E \to \bbR$ is \emph{diminishing return submodular} (\emph{DR-submodular}) if
$
  f(\bix + \bichi_i) - f(\bix) \geq f(\biy + \bichi_i) - f(\biy)
$
for every $\bix \leq \biy$ and $i \in E$, where $\bichi_i$ denotes the $i$-th unit vector.
We note that the lattice submodularity of $f$ does not imply DR-submodularity when the domain is the integer lattice.
Throughout this paper, we assume that $f(\bfzero) = 0$ without loss of generality.

If a function $f:\bbZ^E \to \bbR$ satisfies $f(\bix \vee k\bichi_i)-f(\bix)\geq f(\biy\vee k\bichi_i)-f(\biy)$ for any $i \in E$, $k \in \bbZ_+$, $\bix$ and $\biy$ with $\bix \leq \biy$,
then we say that $f$ satisfies the \emph{weak diminishing return property}.
Any monotone lattice submodular function satisfies the weak diminishing return property~\cite{Soma:2014tp}.

\subsection{Polymatroid}
Let $\rho:2^E \to \bbZ_+$ be a monotone submodular set function with $\rho(\emptyset)=0$.
The (integral) \emph{polymatroid} associated with $\rho$ is the polytope $P = \{\bix\in\bbR_+^E : \bix(X) \leq \rho(X) \quad \forall X \subseteq E\}$, and $\rho$ is called the \emph{rank function} of $P$.
The \emph{base polytope} of polymatroid $P$ is defined as $B := \{\bix \in P : \bix(E) = \rho(E)\}$.
The set of integral points in $B$ satisfies the following \emph{simultaneous exchange property}:
\begin{quote}
    For any $\bix, \biy \in B \cap \bbZ^E_+$ and $s \in \supp^+(\bix-\biy)$, there exists $t \in \supp^+(\biy-\bix)$ such that $\bix - \bichi_s + \bichi_t \in B\cap \bbZ^E_+$ and $\biy + \bichi_s - \bichi_t \in B\cap \bbZ^E_+$.
\end{quote}

The following lemma can be derived by the simultaneous exchange property.

\begin{lemma}\label{lem:base-mapping}
    Let $\bix, \biy \in B \cap \bbZ^E_+$ and $I(\bix) := \{(e, i): e\in \supp^+(\bix), 1\leq i \leq \bix(e)\}$.
    Then there exists a map $\phi: I(\bix)\to \supp^+(\biy)$ such that $\bix - \bichi_e + \bichi_{\phi(e, i)} \in B \cap \bbZ_+^E$ for each $e\in \supp^+(\bix)$ and $1\leq i \leq\bix(e)$, and $\biy = \sum_{(e,i)\in I(\bix)} \bichi_{\phi(e,i)}$.
\end{lemma}
\begin{proof}
    Induction on $|\{\bix\} \setminus \{\biy\}|$.
    If $|\{\bix\} \setminus \{\biy\}| = 0$, then $\phi(e,i):=e$ satisfies the condition.
    Let us assume that $|\{\bix\} \setminus \{\biy\}| > 0$.
    Let us fix $s \in \supp^+(\bix-\biy)$ arbitrarily.
    By the simultaneous exchange property, we can find $t \in \supp^+(\biy-\bix)$ such that $\bix - \bichi_s + \bichi_t \in B\cap \bbZ^E_+$ and $\biy' := \biy + \bichi_s - \bichi_t \in B\cap \bbZ^E_+$.
    By the induction hypothesis, we can obtain $\phi': I(\bix)\to\supp^+(\biy')$ satisfying the conditions.
    The desired $\phi$ can be obtained by modifying $\phi'$ as $\phi'(s, \biy(s)+1) := t$.
    \qed
\end{proof}

\section{Cardinality Constraint}\label{sec:cardinality}

In this section, we consider cardinality constraints.
We provide approximation algorithms for maximizing monotone DR-submodular and lattice submodular functions in Sections~\ref{subsec:cardinality-with-drp} and~\ref{subsec:cardinality-without-drp}, respectively.

\subsection{Maximization of Monotone DR-Submodular Function}\label{subsec:cardinality-with-drp}

We start with the case of a DR-submodular function.
Let $f:\bbZ^E_+ \to \bbR_+$ be a monotone DR-submodular function.
Let $\bic \in \bbZ_+^E$ and $r \in \bbZ_+$.
We want to maximize $f(\bix)$ under constraints $\bfzero \leq \bix \leq \bic$ and $\bix(E)\leq r$.
The pseudocode description of our algorithm, based on the decreasing threshold greedy framework, is shown in Algorithm~\ref{alg:cardinality-with-drp}.

\begin{algorithm}[t!]
  \caption{Cardinality Constraint/DR-Submodular}\label{alg:cardinality-with-drp}
  \begin{algorithmic}[1]
    \REQUIRE{$f:\bbZ_+^E \to \bbR_+$, $\bic\in \bbZ_+^E$, $r \in \bbZ_+$, and $\epsilon > 0$.}
    \ENSURE{$\biy \in \bbZ_+^E$.}
    \STATE{$\biy \leftarrow \bfzero$ and $d \leftarrow \max_{e \in E}f(\bichi_e)$.}
    \FOR{($\theta=d$; $\theta\geq \frac{\epsilon}{r} d$; $\theta \leftarrow \theta(1-\epsilon)$)}
      \FOR{all $e \in E$}
      \STATE{Find maximum $k \leq\min\{\bic(e) - \biy(e),  r - \biy(E)\}$ with $f(k \bichi_e \mid \biy) \geq k\theta$ with binary search.}
        \IF{such $k$ exists}
          \STATE{$\biy \leftarrow \biy + k \bichi_e$.}\label{line:cardinarlity-with-drp-update}
        \ENDIF
      \ENDFOR
    \ENDFOR
    \RETURN{$\biy$.}
  \end{algorithmic}
\end{algorithm}

\begin{lemma}\label{lem:cardinality-update-gain}
  Let $\bix^*$ be an optimal solution.
  When we are adding $k\bichi_e$ to the current solution $\biy$ in Line~\ref{line:cardinarlity-with-drp-update}, the average gain satisfies the following:
  \[
    \frac{f(k\bichi_e \mid \biy)}{k} \geq \frac{(1 - \epsilon)}{r}\sum_{s \in \{\bix^*\} \setminus \{\biy\}}f(\bichi_s \mid \biy).
  \]
\end{lemma}
\begin{proof}
  Due to DR-submodularity, the marginal values can only decrease as we add elements.
  When we are adding a vector $k \bichi_e$ and the current threshold value is $\theta$, the following inequalities hold:
  \begin{claim}
    $f(k\bichi_e \mid \biy)  \geq k\theta $, and $ f(\bichi_s \mid \biy)             \leq \frac{\theta}{1-\epsilon}$ for any $s \in  \{\bix^*\} \setminus \{\biy\}$.
  \end{claim}
  \begin{proof}
    The first inequality is trivial.
    The second inequality is also trivial by DR-submodularity if $\theta = d$.
    Thus we assume that $\theta < d$, i.e., there was at least one threshold update.
    Let $s\in\{\bix^*\} \setminus \{\biy\}$, $k'$ be the increment in the $s$-th entry in the previous threshold (i.e., $\frac{\theta}{1 - \epsilon}$), and $\biy'$ be the variable $\biy$ at the time.
    Suppose that $f(\bichi_s \mid \biy) > \frac{\theta}{1-\epsilon}$.
    Then $f((k'+1)\bichi_s \mid \biy') \geq f(\bichi_s \mid \biy) + f(k'\bichi_s \mid \biy') > \frac{\theta}{1-\epsilon} + \frac{k'\theta}{1-\epsilon} = \frac{(k'+1)\theta}{1-\epsilon}$, which contradicts the fact that $k'$ is the largest value with $f(k'\bichi_s \mid \biy')\geq \frac{k'\theta}{1-\epsilon}$.
    \qed
  \end{proof}

  The above inequalities imply that $\frac{f(k\bichi_e \mid \biy)}{k} \geq (1 - \epsilon)f(\bichi_s \mid \biy)$ for each $s \in \{\bix^*\} \setminus \{\biy\}$.
  Taking the average over these inequalities we obtain
  \[
    \frac{f(k\bichi_e \mid \biy)}{k}
    \geq \frac{1 - \epsilon}{|\{\bix^*\} \setminus \{\biy\}|}\sum_{s \in \{\bix^*\} \setminus \{\biy\}}f(\bichi_s \mid \biy)
    \geq \frac{1 - \epsilon}{r}\sum_{s \in \{\bix^*\} \setminus \{\biy\}}f(\bichi_s \mid \biy)
  \]
  \qed
\end{proof}

\begin{theorem}
  Algorithm~\ref{alg:cardinality-with-drp} achieves an approximation ratio of $1-\frac{1}{e}-\epsilon$ in $O(\frac{n}{\epsilon}\log\|\bic\|_\infty \log\frac{r}{\epsilon})$ time.
\end{theorem}
\begin{proof}
  Let $\biy$ be the output of Algorithm~\ref{alg:cardinality-with-drp}.
  Without loss of generality, we can assume that $\biy(E) = r$.
  To see this, consider a modified version of the algorithm in which the threshold is updated until $\biy(E) = r$.
  Let $\biy'$ be the output of this modified algorithm.
  Since the marginal gain of increasing any coordinate of $\biy$ by one is at most $\epsilon\frac{d}{r}$, $f(\biy') - f(\biy) \leq \epsilon d \leq \epsilon\OPT$.
  Therefore, it suffices to show that $\biy'$ is a $(1-1/e-\epsilon)$-approximate solution.

  Let $\biy_i$ be the vector $\biy$ following the $i$-th update.
  We define $\biy_0 = \bfzero$.
  Let $k_i \bichi_{e_i}$ be the vector added during the $i$-th update.
  That is, $\biy_{i} = \sum_{j=1}^i k_j \bichi_{e_j}$.
  By Lemma~\ref{lem:cardinality-update-gain},
  for any $i \in \bbN$,
  \[
    \frac{f(k_{i}\bichi_{e_{i}} \mid \biy_{i-1})}{k_{i}} \geq \frac{1-\epsilon}{r} \sum_{s \in \{\bix^*\}   \setminus  \{\biy_{i-1}\} }f(\bichi_s \mid \biy_{i-1}).
  \]
  By DR-submodularity, $\sum_{s \in \{\bix^*\} \setminus \{\biy_{i-1}\}} f(\bichi_s \mid \biy_{i-1}) \geq f(\bix^* \vee \biy_{i-1}) - f(\biy_{i-1})$ holds.
  Therefore by monotonicity, 
  \begin{align*}
    f(\biy_{i}) - f(\biy_{i-1}) &= f(k_{i}\bichi_{e_{i}} \mid \biy_{i-1})\\
    & \geq
    \frac{(1-\epsilon)k_{i}}{r} (f(\bix^* \vee \biy_{i-1}) - f(\biy_{i-1}))
    \\
    &\geq
    \frac{(1-\epsilon)k_{i}}{r} (\OPT - f(\biy_{i-1})).
  \end{align*}
  Hence, we can show by induction that $
    f(\biy)
    \geq
    \left(1 - \prod_i \left(1 - \frac{(1-\epsilon)k_i}{r}\right)\right) \OPT.
 $
  Since
 $
    \prod_i \left(1-\frac{(1-\epsilon)k_i}{r} \right)
    \leq
    \prod_i \exp\left(-\frac{(1-\epsilon)k_i}{r} \right)
    =
    \exp\left( -\frac{(1-\epsilon)\sum_i k_i}{r} \right)
    = e^{-(1-\epsilon)}
    \leq \frac{1}{e} + \epsilon,
$
we obtain $(1 - \frac{1}{e} - \epsilon)$-approximation.
\qed
\end{proof}

\subsection{Maximization of Monotone Lattice Submodular Function}\label{subsec:cardinality-without-drp}
\begin{algorithm}[t!]
  \caption{\textsf{BinarySearchLattice}$(f,e,\theta,k_{\max},\epsilon)$}\label{alg:binary-search-cardinality-without-drp}
  \begin{algorithmic}[1]
  \REQUIRE{$f:\bbZ_+^E \to \bbR_+$, $e \in E$, $\theta >0$, $k_{\max} \in \bbZ_+$, $\epsilon > 0$.}
  \ENSURE{$0 \leq k \leq k_{\max}$ or \textbf{fail}.}
  \STATE{Find $k_{\min}$ with $0 \leq k_{\min} \leq k_{\max}$ such that $f(k_{\min}\bichi_e) > 0$ by binary search.}
  \STATE{\textbf{if} no such $k_{\min}$ exists \textbf{then} \textbf{fail}.}
  \FOR{ ($h = f(k_{\max} \bichi_e)$; $h \geq (1-\epsilon)f(k_{\min} \bichi_e) $; $h = (1-\epsilon)h$) }\label{line:for-binary-search-cardinality}
  \STATE{Find the smallest $k$ with $k_{\min} \leq k \leq k_{\max}$ such that $f(k \bichi_e) \geq h$ by binary search.}
  \IF{ $f(k\bichi_e) \geq (1-\epsilon)k\theta$ }
    \RETURN{$k$.}
  \ENDIF
  \ENDFOR
  \STATE{\textbf{fail}.}
  \end{algorithmic}
\end{algorithm}

We now consider the case of a lattice submodular function.
Let $f:\bbZ^E_+ \to \bbR_+$ be a monotone lattice submodular function, $\bic \in \bbZ_+^E$, and $r \in \bbZ_+$.
We want to maximize $f(\bix)$ under the constraints $0\leq \bix \leq \bic$ and $\bix(E)\leq r$.

The main issue is that we cannot find $k$ such that $f(k\bichi_e \mid \bix) \geq k\theta$ by naive binary search.
However, we can find $k$ such that $f(k\bichi_e \mid \bix) \geq (1-\epsilon) k\theta$ in polynomial time (if exists).
The key idea is guessing the value of $f(k \bichi_e)$ by iteratively decreasing the threshold and checking whether the desired $k$ exists with binary search.
See Algorithm~\ref{alg:binary-search-cardinality-without-drp} for details.

We have the following properties:
\begin{lemma}\label{lem:binary-search-cardinality}
  Algorithm~\ref{alg:binary-search-cardinality-without-drp} satisfies the following:
  \begin{itemize}
  \itemsep=0pt
  \item[(1)] Suppose that there exists $0 \leq k^* \leq k_{\max}$ such that $f(k^*\bichi_e) \geq k^*\theta$.
  Then, Algorithm~\ref{alg:binary-search-cardinality-without-drp} returns $k$ with $k_{\min} \leq k \leq k_{\max}$ such that $f(k\bichi_e) \geq (1-\epsilon)k\theta$.
  \item[(2)] Suppose that Algorithm~\ref{alg:binary-search-cardinality-without-drp} outputs $0 \leq k \leq k_{\max}$.
    Then, $f(k'\bichi_e) < \max\{\frac{f(k\bichi_e)}{1-\epsilon}, k'\theta\}$ for any $k < k' \leq k_{\max}$.
  \item[(3)] If Algorithm~\ref{alg:binary-search-cardinality-without-drp} does not fail,
  then the output $k$ satisfies $f(k\bichi_e) \geq (1-\epsilon)k\theta$.
  \item[(4)] Let $k_{\min} = \min\{k \mid f(k\bichi_e) > 0\}$. Then, Algorithm~\ref{alg:binary-search-cardinality-without-drp} runs in  $O(\frac{1}{\epsilon}\log \frac{f(k_{\max}\bichi_e)}{f(k_{\min}\bichi_e)} \cdot \log k_{\max})$ time.
  If no such $k_{\min}$ exists, then Algorithm~\ref{alg:binary-search-cardinality-without-drp} runs in $O(\log k_{\max})$ time.
  \end{itemize}
\end{lemma}
\begin{proof}
  (1)
  Let  $H := \{ (1-\epsilon)^s f(k_{\max}\bichi_e)  : s \in \bbZ_+, (1-\epsilon)^s f(k_{\max}\bichi_e) \geq (1-\epsilon)f(k_{\min}\bichi_e) \}$.
  Let $h^*$ be the (unique) element in $H$ such that $h^* \leq f(k^*\bichi_e) < \frac{h^*}{1-\epsilon}$.
  Let $k$ be the minimum integer such that $f(k\bichi_e) \geq h^*$.
  Then, $k \leq k^*$ and $f(k^*\bichi_e) < \frac{f(k\bichi_e)}{1-\epsilon}$.
  Thus $k\theta \leq k^*\theta \leq f(k^* \bichi_e) \leq \frac{f(k\bichi_e)}{1-\epsilon}$, which means that $(1-\epsilon) k\theta \leq f(k\bichi_e)$.

  (2)
  Let $h$ and $h'$ be the unique elements in $H$ such that $h \leq f(k\bichi_e) < \frac{h}{1-\epsilon}$ and $h' \leq f(k'\bichi_e) < \frac{h'}{1-\epsilon}$, respectively.
  Note that $h \leq h'$.
  If $h = h'$ then $f(k'\bichi_e) \leq \frac{f(k\bichi_e)}{1 - \epsilon} $.
  If $h < h'$, let $k_1$ be the minimum $k_1$ such that $f(k_1\bichi_e) \geq h'$.
  Then $\frac{f(k'\bichi_e)}{k'} \leq \frac{f(k_1\bichi_e)}{(1-\epsilon)k_1} < \theta$, where the last inequality follows from the fact that $k_1$ is examined by the algorithm before $k$.

  (3) and (4) are obvious.
  \qed
\end{proof}

Our algorithm for maximizing monotone lattice submodular functions under a cardinality constraint is based on the decreasing threshold greedy framework and uses Algorithm~\ref{alg:binary-search-cardinality-without-drp} to find elements whose marginal gain is as large as the current threshold.
The pseudocode description is shown in Algorithm~\ref{alg:cardinality}.

\begin{algorithm}[t!]
  \caption{Cardinality Constraint/Lattice Submodular}\label{alg:cardinality}
  \begin{algorithmic}[1]
    \REQUIRE{$f:\bbZ_+^E \to \bbR_+$, $\bic \in \bbZ_+^E$, $r \in \bbZ_+$, $\epsilon > 0$.}
    \STATE{$\biy \leftarrow \bfzero$ and $d_{\max} \leftarrow \max_{e \in E}f(\bic(e)\bichi_e)$}.
    \FOR{($\theta=d_{\max}$; $\theta\geq \frac{\epsilon}{r} d_{\max}$; $\theta \leftarrow \theta(1-\epsilon)$)}
      \FOR{all $e \in E$}
      \STATE{Invoke \textsf{BinarySearchLattice}($f(\cdot \mid \biy), e, \theta, \min\{\bic(e) - \biy(e),  r - \biy(E)\}, \epsilon$)}
        \IF{\textsf{BinarySearchLattice} did not fail and returned $k \in \bbN$}
          \STATE{$\biy \leftarrow \biy + k \bichi_e$.}\label{line:cardinarlity-update}
        \ENDIF
      \ENDFOR
    \ENDFOR
    \RETURN $\biy$.
  \end{algorithmic}
\end{algorithm}

Let $\theta_i$ be $\theta$ in the $i$-th iteration of the outer loop.
Let $k_{i,e}$ be $k$ when handling $e \in E$ in the $i$-th iteration of the outer loop.
Let $\biy_{i,e}$ be the vector $\biy$ right before adding the vector $k_{i,e}\bichi_e$.
For notational simplicity, we define $\biy_{0,e} = \bfzero$ and $k_{0,e} = 0$ for all $e \in E$.
Let $\biDelta_{i,e} = (\bix^* - \biy_{i,e}) \vee \bfzero$.
The following lemma shows that $f(\biDelta_{i,e}(a)\bichi_a \mid \biy_{i,e}) \leq \biDelta_{i,e}(a)\theta_i$ for any $i \in \bbN$ and $e, a \in E$, which is a crucial property to guarantee the approximation ratio of Algorithm~\ref{alg:cardinality}, except that there are several small error terms.

\begin{lemma}\label{lem:cardinality-marginal}
  For any $i \in \bbN$ and $e,a \in E$, 
  \[
    f(\biDelta_{i,e}(a)\bichi_a \mid \biy_{i,e})
    \leq
    \max\Bigl\{
    \frac{\epsilon}{1-\epsilon} f(k_{i-1,a}\bichi_a \mid \biy_{i-1,a}),
    \bigl(\epsilon k_{i-1,a} +  \biDelta_{i,e}(a)\bigr)\frac{\theta_i}{1-\epsilon}
    \Bigr\}.
  \]
\end{lemma}
\begin{proof}
  We define $\biDelta = \biDelta_{i,e}$, $k' = k_{i-1,a}$, $\biy = \biy_{i,a}$, $\biy' = \biy_{i-1,a}$, and $\theta' = \frac{\theta_i}{1-\epsilon}$ for notational simplicity.
  We assume $\biDelta(a) > 0$ as otherwise the statement is trivial.
  When $i \geq 2$, from~(2) in Lemma~\ref{lem:binary-search-cardinality},
  \[
    f((k'+\biDelta(a))\bichi_a \mid \biy')
    \leq \max\Bigl\{\frac{f(k'\bichi_a \mid \biy')}{1-\epsilon}, (k'+\biDelta(a))\theta'\Bigr\}.
  \]
  We can verify that this inequality also holds when $i = 1$ because
  \[
    f((k'+\biDelta(a))\bichi_a \mid \biy') =
    f(\biDelta(a)\bichi_a) \leq d_{\max} \leq (k'+\biDelta(a))\theta'.
  \]

  Since $f((k'+\biDelta(a))\bichi_a \mid \biy') \geq f(\biDelta(a)\bichi_a \mid \biy) + f(k'\bichi_a \mid \biy')$
  holds from lattice submodularity, it follows that
  \begin{align*}
    f(\biDelta(a)\bichi_a \mid \biy) \leq \max\Bigl\{\frac{\epsilon}{1-\epsilon} f(k'\bichi_a \mid \biy'), (k'+\biDelta(a))\theta' - f(k'\bichi_a \mid \biy')\Bigr\}.
  \end{align*}
  Since $f(k'\bichi_a \mid \biy') \geq (1-\epsilon)k'\theta'$ from~(1) of Lemma~\ref{lem:binary-search-cardinality},
  we obtain the claim.
  \qed
\end{proof}

\begin{lemma}\label{lem:cardinality-marginal-sum}
  \[
    f(k_{i,e}\bichi_e \mid \biy_{i,e})
    \geq
    \frac{(1-\epsilon)^2k_{i,e}}{(1+\epsilon)r}\Bigl(\sum_{a \in E}f(\biDelta_{i,e}(a)\bichi_a \mid \biy_{i,e})
    - \frac{\epsilon}{1-\epsilon} f(\biy_{i,e}) \Bigr).
  \]
\end{lemma}
\begin{proof}
  Note that for any $i \in \bbN$ and $e \in E$,
  $f(k_{i,e}\bichi_e \mid \biy_{i,e})  \geq (1-\epsilon)k_{i,e}\theta_{i}$ from~(1) of Lemma~\ref{lem:binary-search-cardinality}.
  Hence, by summing up the inequalities of Lemma~\ref{lem:cardinality-marginal} over all $a \in E$, we obtain
  \begin{align*}
    & \sum_{a \in E}f(\biDelta_{i,e}(a)\bichi_a \mid \biy_{i,e}) \\
    & \leq
    \sum_{a \in E}\bigl(\epsilon k_{i-1,a} +  \biDelta_{i,e}(a)\bigr)\frac{\theta_i}{1-\epsilon} +
    \sum_{a \in E}\frac{\epsilon}{1-\epsilon} f(k_{i-1,a}\bichi_a \mid \biy_{i-1,a}) \\
    & \leq
    (1+\epsilon)r\frac{\theta_i}{1-\epsilon} + \frac{\epsilon}{1-\epsilon} f(\biy_{i,e}) \\
    & \leq
    \frac{(1+\epsilon)r}{(1-\epsilon)^2 k_{i,e}}f(k_{i,e}\bichi_e \mid \biy_{i,e}) +
    \frac{\epsilon}{1-\epsilon} f(\biy_{i,e}).
  \end{align*}
  We obtain the claim by rearranging the inequality.
  \qed
\end{proof}

\begin{theorem}
  Algorithm~\ref{alg:cardinality} achieves an approximation ratio of $1-\frac{1}{e}-O(\epsilon)$ in
  $O(\frac{n}{\epsilon^2}\log\|\bic\|_\infty \log\frac{r}{\epsilon} \log \tau )$ time,
  where $\tau = \frac{ \max_{e \in E} f(c(e)\bichi_e) }{ \min\{ f(\bichi_e \mid \bix) : e \in E, \bfzero \leq \bix \leq \bic, f(\bichi_e \mid \bix) > 0 \}}$.
\end{theorem}
\begin{proof}
  Let $\biy$ be the final output of Algorithm~\ref{alg:cardinality}.
  We can assume that $\biy(E) = r$.
  To see this, consider the modified algorithm in which $\theta$ is updated until $\biy(E) = r$.
  The output $\biy'$ of this modified algorithm satisfies $\biy'(E) = \biy$.
  By a similar argument as above, we can show that
  $\sum_{a \in E} f((\biy' - \biy)(a)\bichi_a \mid \biy)
  \leq (1+\epsilon) r \frac{\theta_{\min}}{1-\epsilon} + \frac{\epsilon}{1-\epsilon} f(\biy)$,
  where $\theta_{\min} = \frac{\epsilon}{r} d_{\max}$.
  This yields $f(\biy') \leq (1 + O(\epsilon))f(\biy) + O(\epsilon) \OPT$.
  Hence it suffices to show that $\biy'$ gives $(1-1/e-\epsilon)$-approximation.

  Let $\alpha = \frac{(1-\epsilon)^2}{1+\epsilon}$ and $\beta = \frac{\epsilon }{1-\epsilon}$.
  By Lemma~\ref{lem:cardinality-marginal-sum} and lattice submodularity of $f$,
  \begin{align*}
    f(k_{i,e}\bichi_e \mid \biy_{i,e})
    & \geq
    \frac{\alpha k_{i,e}}{r}\Bigl(\sum_{a \in E}f(\biDelta_{i,e}(a)\bichi_a \mid \biy_{i,e})
    - \beta f(\biy_{i,e}) \Bigr)\\
    & \geq
    \frac{\alpha k_{i,e}}{r}\Bigl((f(\bix^* \vee \biy_{i,e}) - f(\biy_{i,e}))
    - \beta  f(\biy_{i,e})\Bigr) \\
    & \geq
    \frac{\alpha k_{i,e}}{r} \Bigl(\OPT - (1 + \beta) f(\biy_{i,e}) \Bigr) \\
& = \alpha(1+\beta)\frac{k_{i,e}}{r} \left( \widetilde{\OPT} - f(\biy_{i,e})\right) \\
    & = \frac{(1- O(\epsilon))k_{i,e}}{r} \left( \widetilde{\OPT} - f(\biy_{i,e})\right),
  \end{align*}
  where $\widetilde{\OPT} = \OPT / (1+ \beta) = (1 - \epsilon)\OPT$.
  By the same argument in the case of DR-submodular functions, we can show that $f(\biy) \geq (1 - 1/e - O(\epsilon))\widetilde{\OPT} = (1 - 1/e - O(\epsilon))\OPT$.
  The analysis of the time complexity is straightforward.
  \qed
\end{proof}

\section{Polymatroid Constraint}\label{sec:polymatroid}
Let $P$ be a polymatroid with a ground set $E$ and the rank function $\rho:2^E\to\bbZ_+$.
The objective is to maximize $f(\bix)$ subject to $\bix\in P\cap\bbZ^E$, where $f$ is a DR-submodular function. In what follows, we denote $\rho(E)$ by $r$.
We assume that $P$ is contained in the interval $[\bfzero, \bic]:=\{\bix\in\bbR^E_+ : \bfzero \leq \bix \leq \bic \}$.

We start by describing a continuous extension for polymatroid constraints in Section~\ref{subsec:continuous-extension-polymatroid} and then state and analyze our algorithm in Section~\ref{subsec:algorithm-polymatroid}.

\subsection{Continuous extension for polymatroid constraints}\label{subsec:continuous-extension-polymatroid}
For $\bix \in \bbR^E$, let $\lfloor \bix \rfloor$ denote the vector obtained by rounding down each entry of $\bix$.
For $a\in\bbR$, let $\langle a \rangle$ denote the fractional part of $a$,
that is,
$\langle a \rangle := a - \lfloor a \rfloor$.
For $\bix \in \bbR^E$,
we define $C(\bix) := \{ \biy \in \bbR^E \mid \lfloor \bix \rfloor \leq \biy \leq \lfloor \bix \rfloor + \mathbf{1} \}$ as the hypercube to which $\bix$ belongs.

For $\bix \in \bbR^E$,
we define $\caD(\bix)$ as the distribution from which we sample $\bar{\bix}$ such that $\bar{\bix}(i) = \lfloor \bix(i) \rfloor$ with probability $1- \langle \bix(i)\rangle$ and $\bar{\bix}(i) = \lceil \bix (i)\rceil $ with probability $\langle \bix(i) \rangle$, for each $i\in E$.
We define the continuous extension $F:\bbR_+^E \to \bbR_+$ of $f:\bbZ_+^E \to \bbR_+$ as follows.
For each $\bix \in \bbR_+^E$,
we define
\begin{align}\label{eq:ext_polym}
  F(\bix)
  := \E_{\bar{\bix} \sim \caD(\bix)} [f(\bar{\bix})]
  = \sum_{S \subseteq E}f(\lfloor \bix\rfloor + \bichi_S) \prod_{i \in S}\langle \bix(i) \rangle \prod_{i \not \in S}(1 - \langle \bix(i) \rangle).
  \end{align}
  We call this type of continuous extension \emph{the continuous extension for polymatroid constraints}.
Note that $F$ is obtained by gluing the multilinear extension of $f$ restricted to each hypercube.
If $f:\bit^E \to \bbR_+$ is a monotone submodular function,
it is known that its multilinear extension is monotone and concave along non-negative directions.
We can show similar properties for the continuous extension of a function $f:\bbZ_+^E \to \bbR_+$ if $f$ is monotone and DR-submodular.

\begin{lemma}\label{lem:monotone+concave}
    For a monotone DR-submodular function $f$, the continuous extension $F$ for the polymatroid constraint is a non-decreasing concave function along any line of direction $\bid\geq 0$.
\end{lemma}
To prove this, we need the following useful fact from convex analysis.

\begin{lemma}[Rockafellar~\cite{Rockafellar1996}, Theorem 24.2]\label{lem:Rocka}
Let $g:\bbR\to\bbR$ be a non-increasing function and $a \in\bbR$. Then
\[
f(x) := \int_a ^ x g(t)dt
\]
is a concave function.
\end{lemma}
We will use the following notation from~\cite{Rockafellar1996}; For a function $\phi:\bbR\to\bbR$ and $a\in\bbR$, let us define
\begin{align*}
\phi'_+(a) := \lim_{\epsilon\downarrow 0}\frac{\phi(a + \epsilon) - \phi(a)}{\epsilon} \quad \text{and} \quad
\phi'_-(a) := \lim_{\epsilon\uparrow0}\frac{\phi(a + \epsilon) - \phi(a)}{\epsilon},
\end{align*}
if the limits exist.
For a multivariate function $F:\bbR^E \to\bbR$, $\nabla_+ F$ and $\nabla_- F$ are defined in a similar manner. 
For $i \in E$, the $i$-th entry of $\nabla_+ F$ equals $\phi'_+$ where $\phi = \frac{\partial F}{\partial x_i}$ and $\nabla_- F$ is defined accordingly.

\begin{proof}[of Lemma~\ref{lem:monotone+concave}]
Let $\bip \in \bbR^E_+$ and $\phi(\xi) := f(\bip + \xi \bid)$ for $\xi\geq 0$. To show that $\phi$ is a non-decreasing concave function, we will find a non-increasing non-negative function $g$ such that $\phi(\xi) = \phi(0) + \int_0^\xi g(t)dt$. Note that if $\phi$ is differentiable for all $\xi\geq 0$, then $g:=\phi'$ satisfies the condition. However, $\phi$ may be non-differentiable at $\xi$ if $\bip+\xi\bid$ contains an integral entry.

Let us denote $\bix:= \bip+\xi_0\bid$ and suppose that $x_i \in \bbZ$ for some $i$.
Then one can check that $\phi_-'(\xi_0) = (\nabla F_- \rvert_{\xi_0} )^\top \bid$ and $\phi_+'(\xi_0) = (\nabla F_+ \rvert_{\xi_0})^\top \bid$.
For each $j$,
  \begin{align*}
      \nabla F_-(j) \rvert_{\xi_0} & = \sum_{S: i \not \in S}f( \bichi_i \mid \lfloor \bix \rfloor - \bichi_i + \bichi_S)\prod_{j \in S} \langle \bix(j) \rangle \prod_{j \not \in S, j \neq i}(1 - \langle \bix(j) \rangle), \\
      \nabla F_+(j) \rvert_{\xi_0}  & = \sum_{S: i \not \in S}f(\bichi_i \mid \lfloor \bix \rfloor + \bichi_S)\prod_{j \in S} \langle \bix(j) \rangle \prod_{j \not \in S, j \neq i}(1 - \langle \bix(j) \rangle).
  \end{align*}
On the other hand, from DR-submodularity, 
  \[
     f(\bichi_i \mid \lfloor \bix \rfloor - \bichi_i +  \bichi_{S}) \geq f(\bichi_i \mid \lfloor \bix \rfloor + \bichi_S).
  \]
Since $\bid\geq\bfzero$,  we have $\phi'_-(\xi_0) \geq \phi'_+(\xi_0)$.

Let $g:\bbR_+ \to \bbR$ be an arbitrary function such that $\phi'_-(\xi)\geq g(\xi)\geq \phi'_+(\xi)$ for $\xi\geq 0$. Note that $g(\xi) = \phi'(\xi)$ if $\phi$ is differentiable at $\xi$. Then $g$ is non-decreasing and non-negative.
Since $g(\xi) = \phi'(\xi)$ almost everywhere, $\phi(\xi) = \phi(0) + \int_0^\xi g(t)dt$. The lemma now directly follows from the non-negativity of $g$ and Lemma~\ref{lem:Rocka}.
\qed
\end{proof}

For a continuous extension $F$ for polymatroid constraints and $\bix,\biy \in \bbR^E$, we define $F(\bix \mid \biy) = F(\bix + \biy) - F(\biy)$.
The following is immediate.
\begin{proposition}
  If $\bix \in \bbZ_+^E$, then $F(\bix \mid \biy) = \E_{\biz \sim \caD(\biy)}[f(\bix \mid \biz)]$.
\end{proposition}

\subsection{Continuous greedy algorithm for polymatroid constraints}\label{subsec:algorithm-polymatroid}
In this section, we describe and analyze our algorithm for the polymatroid constraint, whose pseudocode description is presented in Algorithm~\ref{alg:Accelerated-Cont-Greedy}.

At a high level, our algorithm computes a sequence $\bix^0 = \bfzero, \dots, \bix^{1/\epsilon}$ in $P$ (recall that we have assumed that $1/\epsilon$ is an integer).
For each time step $t \in [1/\epsilon]$, given $\bix^{t-1}$, we determine a direction $\biy^t$ by calling a subroutine \textsf{DirectionPolymatroid}, and we update as $\bix^{t} := \bix^{t-1} + \epsilon \biy^t$.
Intuitively speaking, $\biy^t$ is the direction that maximizes the marginal gain while keeping $\bix^t$ in $P$.
Finally, we execute a rounding algorithm \textsf{RoundingPolymatroid} to the fractional solution $\bix^{1/\epsilon}$, and obtain an integral solution $\bar{\bix}$.
The detailed description on \textsf{DirectionPolymatroid} and \textsf{RoundingPolymatroid} are given in Sections~\ref{subsubsec:gradient-polymatroid} and~\ref{subsubsec:rounding-polymatroid}, respectively.

\begin{algorithm}[t]
    \begin{algorithmic}[1]
  \caption{Polymatroid Constraint/DR-Submodular}\label{alg:Accelerated-Cont-Greedy}
  \REQUIRE{$f:\bbZ_+^E \to \bbR_+, P \subseteq \bbR_+^E$, and $\epsilon > 0$.}
  \ENSURE{A vector $\bar{\bix} \in P \cap \bbZ_+^E$.}
  \STATE{$\bix^0 \leftarrow \bfzero$.}
  \FOR{($t \leftarrow 1$; $t \leq \lfloor \frac{1}{\epsilon}\rfloor$; $t\leftarrow t+1$)}
    \STATE{$\biy^t \leftarrow \textsf{DirectionPolymatroid}(f, \bix^{t-1}, \epsilon, P)$.}
    \STATE{$\bix^t \leftarrow \bix^{t-1} + \epsilon \biy^t$.}
  \ENDFOR
  \STATE{$\bar{\bix} \leftarrow \textsf{RoundingPolymatroid}(\bix^{1/\epsilon}, P)$.}
  \RETURN{$\bar{\bix}$.}
  \end{algorithmic}
\end{algorithm}


\subsubsection{Computing an approximation to the gradient}\label{subsubsec:gradient-polymatroid}

\begin{algorithm}[t]
  \caption{\textsf{DirectionPolymatroid}$(f,\bix,\epsilon,P)$}\label{alg:polymatroid-decreasing-threshold}
  \begin{algorithmic}[1]
  \REQUIRE{$f:\bbZ_+^E \to \bbR_+$, $\bix\in \bbR_+^E$, $\epsilon\in [0,1]$, $P\subseteq \bbR_+^E$.}
  \ENSURE{A vector $\biy \in P\cap \bbZ_+^E$.}
  \STATE{$\biy \leftarrow \bfzero$ and $d \leftarrow \max_{e\in E} f(\bichi_e)$. }
  \STATE{$N \leftarrow $ the solution to $N = n \lceil \log_{1/1-\epsilon}\frac{N}{\epsilon} \rceil$. Note that $N = O(\frac{n}{\epsilon}\log \frac{n}{\epsilon})$.}
  \FOR{($\theta=d$; $\theta \geq  \frac{\epsilon d}{N}$; $\theta \leftarrow \theta(1- \epsilon))$}
    \FOR{all $e \in E$}
      \STATE{$k_{\max} \leftarrow \max\{k:\bix + \biy + k\bichi_e \in P \}$.}
      \STATE{$k \leftarrow$ \textsf{BinarySearchPolymatroid}$(f,\bix,e,\theta,\alpha_{\ref{alg:polymatroid-decreasing-threshold}},\beta_{\ref{alg:polymatroid-decreasing-threshold}},\delta_{\ref{alg:polymatroid-decreasing-threshold}},k_{\max})$.
      }
      \IF{$k \geq 1$}
        \STATE{$\biy \leftarrow \biy + k\bichi_e$.}
      \ENDIF
    \ENDFOR
  \ENDFOR
  \RETURN $\biy$.
  \end{algorithmic}
\end{algorithm}
\begin{algorithm}[t]
  \caption{\textsf{BinarySearchPolymatroid}$(f,\bix,e,\theta,\alpha,\beta,\delta,k_{\max})$}\label{alg:binary-search-polymatroid}
  \begin{algorithmic}[1]
      \REQUIRE{$f:\bbZ_+^E \to \bbR_+$, $\bix \in \bbR_+^E$, $e \in E$, $\theta \in \bbR_+$, $\alpha \in (0,1/2)$, $\beta,\delta \in (0,1)$, $k_{\max} \in \bbZ_+$}
    \ENSURE{$k \in \bbZ_+$.}
    \STATE{$\ell \leftarrow 1$, $u \leftarrow k_{\max}$.}
    \WHILE{$\ell < u$}
    \STATE{$m = \lfloor \frac{\ell + u}{2} \rfloor$.}
    \STATE{$\widetilde{F}(m\bichi_e \mid \bix) \leftarrow$ the estimate of $F(m\bichi_e \mid \bix)$ obtained by averaging $O\left(\frac{\log (k_{\max} /\delta)}{\alpha \beta}\right)$ random samples.}
    \IF{$\widetilde{F}(m\bichi_e \mid \bix) \geq m\theta$}
    \STATE{$\ell \leftarrow m + 1$.}\label{line:binary-search-update-l}
    \ELSE
    \STATE{$u \leftarrow m$.}\label{line:binary-search-update-r}
    \ENDIF
    \ENDWHILE
    \STATE{$k \leftarrow \ell - 1$.}
    \RETURN $k$.
  \end{algorithmic}
\end{algorithm}

As mentioned above, given a vector $\bix \in P$, we want to compute a direction $\biy$ that maximizes the marginal gain $F(\epsilon \biy \mid \bix)$ and satisfies $\bix + \epsilon \biy \in P$.
In order to efficiently approximate such $\biy$, again we use the decreasing-threshold greedy framework.
The pseudocode of \textsf{DirectionPolymatroid} is given in Algorithm~\ref{alg:polymatroid-decreasing-threshold}.
The parameters $\alpha_{\ref{alg:polymatroid-decreasing-threshold}}$, $\beta_{\ref{alg:polymatroid-decreasing-threshold}}$, and $\delta_{\ref{alg:polymatroid-decreasing-threshold}}$ will be determined later.

Recall that, in the framework, we want to determine how many copies of an element can be added when the average marginal gain is restricted to be as large as the current threshold.
Given $\bix \in \bbR_+^E$, $e \in E$, and $\theta \in \bbR_+$, \textsf{BinarySearchPolymatroid} (Algorithm~\ref{alg:binary-search-polymatroid}) performs binary search to find the largest $k$ such that $F(k\bichi_e \mid \bix) \geq \theta$.
Since we can only estimate the value of $F(\cdot)$ in polynomial time, the output $k$ of Algorithm~\ref{alg:binary-search-polymatroid} only satisfies the following weak conditions:
\begin{lemma}\label{lem:binary-search-polymatroid}
  Algorithm~\ref{alg:binary-search-polymatroid} satisfies the following properties:
  \begin{itemize}
  \item[(1)] Suppose that Algorithm~\ref{alg:binary-search-polymatroid} outputs $k \in \bbZ_+$.
    Then, with probability at least $ 1- \delta$,
    \[
      F(k\bichi_e \mid \bix)  \geq (1-\alpha)k\theta - \beta f(k \bichi_e)
    \]
    and
    \[
      F((k+1)\bichi_e \mid \bix)  < \frac{(k+1)\theta}{1-\alpha} + 2\beta f((k+1)\bichi_e).
    \]
  \item[(2)] Suppose that Algorithm~\ref{alg:binary-search-polymatroid} outputs $k \in \bbZ_+$.
    Then, with probability at least $ 1- \delta$, for any $k < k' \leq k_{\max}$,
    \[
      F(k'\bichi_e \mid \bix)  < \frac{k'\theta}{1-\alpha} + 2\beta f(k_{\max}\bichi_e).
    \]
\item[(3)] Algorithm~\ref{alg:binary-search-polymatroid} runs in $O(\log k_{\max} \cdot \frac{\log (k_{\max}/\delta)}{\alpha \beta} )$ time.
  \end{itemize}

\end{lemma}
\begin{proof}
  (1)
  For $k \in \bbZ$,
  let $s_m = F(m\bichi_e \mid \bix) = \E_{\biz \sim \caD(\bix)}[f(m\bichi_e \mid \biz)]$ and $\widetilde{s}_m$ be its estimation.
  Note that $0 \leq f(m\bichi_e \mid \biz) \leq f( m \bichi_e)$ for any $\biz$ in the support of $\caD(\bix)$.
  By Lemma~\ref{lem:chernoff} and the union bound,
  with probability at least $1 - \delta$,
  \begin{align*}
    |\widetilde{s}_{m} - s_{m}| & \leq \alpha s_{m} + \beta f( m \bichi_e).
  \end{align*}
  for every estimation in the process.
  In what follows, we suppose this happens.

  Suppose $s_{m} < \frac{1}{1+\alpha}(m \theta- \beta f( m \bichi_e))$.
  Then, $\widetilde{s}_{m} < m \theta$ and we reach Line~\ref{line:binary-search-update-r}.
  On the other hand,
  Suppose $s_{m} \geq \frac{1}{1-\alpha}(m \theta + \beta f( m \bichi_e))$.
  Then, $\widetilde{s}_{m} \geq m \theta$ and we reach Line~\ref{line:binary-search-update-l}.
  Hence after the while loop,
  \begin{align*}
    s_{\ell-1} & \geq \frac{1}{1+\alpha} ((\ell-1) \theta - \beta f( (\ell-1) \bichi_e)) \geq (1-\alpha)(\ell-1) \theta - \beta f( (\ell-1) \bichi_e) \\
    s_{\ell} & < \frac{1}{1-\alpha} (\ell \theta +\beta f( \ell \bichi_e)) = \frac{\ell \theta}{1-\alpha}  + \frac{\beta f( \ell \bichi_e)}{1-\alpha}
    \leq \frac{\ell \theta}{1-\alpha}  + 2\beta f( \ell \bichi_e),
  \end{align*}
  where the last inequality uses the fact that $\alpha < 1/2$.
  Hence, we obtain~(1), (2) holds due to (1) and the concavity of $F$, and (3) is obvious.
  \qed
\end{proof}


\begin{lemma}\label{lem:one-step-improvement}
  Let $\bix^*$ be an optimal solution.
  \textsf{DirectionPolymatroid}$(f,\bix,\epsilon,P)$ produces a vector $\biy \in \bbR_+^E$ such that $\bix' := \bix+\biy$ satisfies
  \[
    F(\bix') - F(\bix) \geq \epsilon\Bigl((1 - 5\epsilon)f(\bix^*) - F(\bix')\Bigr)
  \]
  with probability at least $1-\epsilon/3$.
\end{lemma}
\begin{proof}
  We choose $\alpha_{\ref{alg:polymatroid-decreasing-threshold}} = \epsilon$,
  $\beta_{\ref{alg:polymatroid-decreasing-threshold}} = \frac{\epsilon}{2N(n+1)}$,
  and $\delta_{\ref{alg:polymatroid-decreasing-threshold}} = \frac{\epsilon}{3N}$.
  With probability at least $1-\epsilon/3$,
  all the binary searches succeed.
  In what follows, we assume that this has happened.
  Assume for now that \textsf{DirectionPolymatroid} returns $\biy$ with $\biy(E) = r$.
  If not, we add dummy elements of value $0$ so that $\biy(E) = r$.


  Let $\biy_i$ be the vector $\biy$ after the $i$-th update in the execution of \textsf{DirectionPolymatroid}.
  We define $\biy_0 = \bfzero$.
  Let $b_i$ and $k_i$ be the element of $E$ and the step size chosen in the $i$-th update, respectively; i.e., $\biy_i = \biy_{i-1} + k_i\bichi_{b_i}$.
  Applying Lemma~\ref{lem:base-mapping} to $\biy$ and $\bix^*$, we obtain a mapping $\phi:I(\biy)\to\supp^+(\bix^*)$.
  Let $e^*_{i,j} := \phi(b_i, \biy_{i-1}(b_i) + j)$ for $j = 1, \dots, k_i$.
  Then, by construction, $\biy_{i-1} - \bichi_{b_i} + \bichi_{e^*_{i,j}} \in P$ for $j = 1, \dots, k_i$.

  For each $i \in \bbN$, let $\theta_i$ be the threshold used in the $i$-th update and let $\bix_i = \bix + \biy_{i}$.
  By~(1) of Lemma~\ref{lem:binary-search-polymatroid}, 
  \begin{align}
    F(k_i\bichi_{b_i}\mid\bix_{i-1}) \geq (1-\epsilon)k_i \theta_i - \beta_{\ref{alg:polymatroid-decreasing-threshold}} f(k_i\bichi_{b_i}) - \frac{\epsilon d}{N},
    \label{eq:gradient-polymatroid-1}
  \end{align}
  where the last term $\frac{\epsilon d}{N}$ comes from the case that $b_i$ is a dummy element.

  For each $i \in \bbN$ and $e \in E$, we define $k^*_{i,e} \in \bbZ_+$ as the number of the occurrence of $e$ in the sequence $e^*_{i,1},\ldots,e^*_{i,k_i}$.
  When the vector $k_i\bichi_{b_i}$ is added, any $e \in E$ has been a candidate element when the threshold was $\frac{\theta_i}{1-\epsilon}$ (This is a valid argument only when $i > 1$.
  We consider the case of $i = 1$ presently).
  Thus by~(2) of Lemma~\ref{lem:binary-search-polymatroid},
  for each $e \in E$,
  \[
    F(k^*_{i,e}\bichi_{e} \mid\bix_{i-1}) \leq \frac{k^*_{i,e}\theta_{i}}{(1-\epsilon)^2} + 2\beta_{\ref{alg:polymatroid-decreasing-threshold}} f(k_{\max}\bichi_{e}).
  \]
  This inequality also holds when $i = 1$ because $F(k^*_{i,e}\bichi_{e} \mid\bix_{i-1}) = F(k^*_{i,e}\bichi_{e}) \leq k^*_{i,e}d \leq k^*_{i,e}\theta_{1}$.

  Rephrasing the inequality yields
  \[
    \theta_i \geq \frac{(1-\epsilon)^2}{k^*_{i,e}} \Bigl( F(k_{i,e}^* \bichi_e\mid \bix_{i-1}) - 2\beta_{\ref{alg:polymatroid-decreasing-threshold}} f(k_{\max}\bichi_{e}) \Bigr).
  \]
  By averaging over $e \in E$, we have
  \begin{align}
    \theta_i \geq \frac{(1-\epsilon)^2}{k_i}\sum_e \Bigl( F(k_{i,e}^* \bichi_e\mid \bix_{i-1}) - 2\beta_{\ref{alg:polymatroid-decreasing-threshold}} f(k_{\max}\bichi_{e}) \Bigr).
    \label{eq:gradient-polymatroid-2}
  \end{align}

  Combining~\eqref{eq:gradient-polymatroid-1},~\eqref{eq:gradient-polymatroid-2}, and the fact that $f(\bix^*) \geq d$, we obtain
  \begin{align*}
    & F(k_i \bichi_{b_i} \mid \bix_{i-1}) \\
    & \geq (1 - \epsilon)^3 \sum_{e} \Bigl(F(k^*_{i,e}\bichi_{e} \mid \bix_{i-1} ) - 2\beta_{\ref{alg:polymatroid-decreasing-threshold}} f(k_{\max}\bichi_{e})\Bigr) - \beta_{\ref{alg:polymatroid-decreasing-threshold}} f(k_i \bichi_{b_i}) - \frac{\epsilon d}{N} \\
    & = (1 - \epsilon)^3 \sum_{e} F(k^*_{i,e}\bichi_{e} \mid \bix_{i-1} )- 2\beta_{\ref{alg:polymatroid-decreasing-threshold}} \Bigl( f(k_i \bichi_{b_i}) + \sum_{e}f(k_{\max}\bichi_{e})\Bigr) - \frac{\epsilon d}{N}\\
    & \geq (1 - \epsilon)^3 \sum_{e} F(k^*_{i,e}\bichi_{e} \mid \bix_{i-1} )- 2\beta_{\ref{alg:polymatroid-decreasing-threshold}} \Bigl( f(\bix^*) + \sum_{e}f(\bix^*)\Bigr)- \frac{\epsilon d}{N} \\
    & \geq (1 - \epsilon)^3 \sum_{e} F(k^*_{i,e}\bichi_{e} \mid \bix_{i-1} )- \frac{2\epsilon}{N} f(\bix^*)
  \end{align*}
  Since $F$ is concave along non-negative directions,
  for any $\epsilon > 0$,
  \begin{align}
    F(\epsilon k_i \bichi_{b_i} \mid \bix_{i-1}) \geq
    \epsilon(1 - \epsilon)^3 \sum_{e} F(k^*_{i,e}\bichi_{e} \mid \bix_{i-1} )- \frac{2\epsilon^2}{N}f(\bix^*).
    \label{eq:gradient-polymatroid-3}
  \end{align}
  Using the above inequality and the fact that $N$ is an upper bound on the total number of updates, we bound the improvement at each time step as follows:
  \begin{align*}
    & F(\bix') - F(\bix)
    = \sum_i (F(\bix+\epsilon \biy_i) - F(\bix+\epsilon \biy_{i-1}))
    = \sum_i F(\epsilon k_i\bichi_{b_i} \mid \bix_{i-1}) \\
    & \geq
    \sum_i \left(\epsilon(1 - \epsilon)^3 \sum_{e \in E} F(k^*_{i,e}\bichi_{e} \mid \bix_{i-1} )- \frac{2\epsilon^2}{N}f(\bix^*) \right) \tag{By~\eqref{eq:gradient-polymatroid-3}}\\
    & =
    \epsilon(1 - \epsilon)^3 \sum_i\sum_{e \in E} \E_{\biz \sim \caD(\bix_{i-1})}[f(\biz + k^*_{i,e}\bichi_{e}) - f(\biz)]- \sum_{i}\frac{2\epsilon^2}{N} f(\bix^*)  \\
    & \geq \epsilon (1-\epsilon)^3\E_{\biz \sim \caD(\bix')}[f(\biz \vee \bix^*) - f(\biz)]  - 2\epsilon^2 f(\bix^*)  \tag{from DR-submodularity}\\
    & \geq \epsilon \Bigl((1-\epsilon)^3 (f(\bix^*) - F(\bix')) - 2\epsilon f(\bix^*)\Bigr).  \tag{from monotonicity}\\
    & \geq \epsilon \Bigl((1-5\epsilon)f(\bix^*) - F(\bix')\Bigr).
  \end{align*}
    \qed
\end{proof}

\begin{lemma}\label{lem:polymatroid-one-step-time-complexity}
    \textsf{DirectionPolymatroid} runs in time $\widetilde{O}(\frac{n^2N}{\epsilon^3} \log \frac{n}{\epsilon} \log r \log\frac{rN}{\epsilon})$.
\end{lemma}
\begin{proof}
    Algorithm~\ref{alg:binary-search-polymatroid} takes $O(\frac{nN}{\epsilon^2 } \log r \log \frac{r N}{\epsilon} )$ time.
  The outer loop iterates $\widetilde{O}(\frac{1}{\epsilon}\log\frac{n}{\epsilon})$ times, while the inner loop contains the execution of Algorithm~\ref{alg:binary-search-polymatroid} $n$ times.
  \qed
\end{proof}

\begin{lemma}
  At the end of Algorithm~\ref{alg:Accelerated-Cont-Greedy}, $F(\bix)\geq(1-1/e-O(\epsilon))\OPT$ with probability at least $2/3$.
  Moreover the time complexity is $O(\frac{n^3}{\epsilon^5} \log^3 \frac{n}{\epsilon} \log^2 r)$.
\end{lemma}
\begin{proof}
  Define $\Omega := (1 - 5\epsilon)f(\bix^*) = (1 - 5\epsilon)\OPT$. Let $\bix^t$ be the variable $\bix$ after the $t$-th update.
  Substituting this into the result of Lemma~\ref{lem:one-step-improvement},
  for any $t \in [1/\epsilon]$, we obtain
  \[
    F (\bix^{t}) - F (\bix^{t-1}) \geq \epsilon(\Omega - F (\bix^{t})).
  \]
  Rephrasing the equation, we have
  \[
    \Omega - F(\bix^{t}) \leq \frac{\Omega - F(\bix^{t-1})}{1+\epsilon}.
  \]
  Now applying induction to this equation, we obtain $\Omega  -  F (\bix^t)  \leq \Omega/(1+\epsilon)^{t}$.
  Substituting $t = 1 / \epsilon$ and rewriting the equation we get the desired approximation ratio:
  \begin{align*}
  F (\bix) \geq \left( 1 - \frac{1}{(1+\epsilon)^{1/\epsilon}} \right)\Omega
    \geq (1 -  \frac{1}{e} )(1 - 5\epsilon)\OPT
    =  \Bigl(1  -  \frac{1}{e}  -  O(\epsilon)\Bigr)\OPT.
  \end{align*}

  By Lemma~\ref{lem:polymatroid-one-step-time-complexity}, the total time complexity is $O(\frac{n^2N}{\epsilon^4} \log \frac{n}{\epsilon} \log r \log\frac{rN}{\epsilon}) = O(\frac{n^3}{\epsilon^5} \log^3 \frac{n}{\epsilon} \log^2 r)$.
  \qed
\end{proof}

\subsubsection{Rounding}\label{subsubsec:rounding-polymatroid}

We need a rounding procedure that takes a real vector $\bix^{1/\epsilon}$ as the input and returns an integral vector $\bar{\bix}$ such that $\E [f(\bar{\bix})] \geq F(\bix^{1/\epsilon})$.
There are several rounding algorithms in the $\{0,1\}^E$ case~\cite{Calinescu2011,Chekuri2010}. 
However, generalizing these rounding algorithms over integer lattice is a non-trivial task.
Here, we show that rounding in the integer lattice can be reduced to rounding in the $\{0,1\}^E$ case.

Suppose that we have a fractional solution $\bix$. The following lemma implies that we can round $\bix$ by considering the corresponding matroid polytope.

\begin{lemma}
For $\bix \in P$, $P\cap C(\bix)$ is a translation of a matroid polytope.
\end{lemma}
\begin{proof}
Let us consider a polytope $P' := \{ \langle \biz \rangle : \biz \in P \cap C(\bix) \}$. We can check that $P'$ can be obtained by translating $P\cap C(\bix)$ by $-\lfloor \bix \rfloor$ and restricting to $[0,1]^E$. Therefore, $P' = \{ \biz \in [0,1]^E : \biz(X) \leq \rho'(X) \quad \forall X \subseteq E \}$, where $\rho'(X) := \min_{Y\subseteq X}\{(\rho - \lfloor \bix \rfloor)(Y) + \mathbf{1}(E\setminus Y) \}$. Then we can show that $\rho'$ is the rank function of a matroid by checking the axiom.
\qed
\end{proof}

The independence oracle of the corresponding matroid is simply the independence oracle of $P$ restricted to $C(\bix)$. Thus, the pipage rounding algorithm for $P' = \{ \langle \biz \rangle : \biz \in P \cap C(\bix^{1/\epsilon}) \}$ yields an integral solution $\bar{\bix}$ with $\E [f(\bar{\bix})] \geq F(\bix^{1/\epsilon})$ in strongly polynomial time.

Slightly faster rounding can be achieved by swap rounding. Swap rounding requires that the given fractional solution $\bix$ is represented by a convex combination of extreme points of the matroid polytope. In our setting,
we can represent $\bix^{1/\epsilon}$ as a convex combination of extreme points of $P\cap C(\bix^{1/\epsilon})$ using the algorithm of Cunningham~\cite{Cunningham1984}. Then, we run the swap rounding algorithm for the convex combination and $P\cap C(\bix^{1/\epsilon})$.
The running time of this rounding algorithm is dominated by the complexity of finding a convex combination for $\bix^{1/\epsilon}$, which is $O(n^8)$ time.
Adopting this algorithm as \textsf{RoundingPolymatroid} in Algorithm~\ref{alg:Accelerated-Cont-Greedy}, we get the following:
\begin{theorem}
  Algorithm~\ref{alg:Accelerated-Cont-Greedy} finds an $(1-1/e-\epsilon)$-approximate solution (in expectation) with probability at least $2/3$ in $O(\frac{n^3}{\epsilon^5} \log^3 \frac{n}{\epsilon} \log^2 r + n^8)$ time.
\end{theorem}


\section{Knapsack Constraint}\label{sec:knapsack}

In this section, we give an efficient approximation algorithm for maximizing a DR-submodular function over the integer lattice under a knapsack constraint.
The problem we study is formalized as follows:
given a monotone DR-submodular function  $f:\bbZ_+^E \to \bbR_+$, $\bic \in \bbZ_+^E$, and $\biw \in (0,1]^E$, we want to maximize $f(\bix)$ subject to $\bfzero \leq \bix \leq \bic$ and $\biw^\top \bix \leq 1$.

Our algorithm is twofold.
The basic idea is similar to algorithms for cardinality constraint:
we increase the current solution in a greedy manner using the decreasing threshold greedy framework.
A difference is that the algorithm takes its initial solution as an input, whereas the algorithm for cardinality constraints always uses the zero vector as the initial solution.
This greedy procedure is presented in Section~\ref{sec:knapsack-greedy}.
Obviously, the quality of the output of this greedy procedure depends on the choice of the initial solution.
Here, we use \emph{partial enumeration}, i.e., we try polynomially many initial solutions and return the best one among the outputs of the greedy procedure. This partial enumeration algorithm is described in Section~\ref{sec:knapsack-partial-enumeration}.
The entire algorithm is presented in Section~\ref{sec:knapsack-final}.
Thoughout this section, $\bix^*$ denotes an optimal solution.

\subsection{Greedy Procedure with Decreasing Threshold}\label{sec:knapsack-greedy}
\begin{algorithm}[t!]
    \caption{\textsf{GreedyKnapsack}$(f, \bic, \biw, \bix_0, \epsilon)$}\label{alg:knapsack-greedy}
  \begin{algorithmic}[1]
      \REQUIRE{$f:\bbZ_+^E \to \bbR_+$, $\bic \in \bbZ_+^E$, $\biw \in (0,1]^E$, $\bix_0 \in \bbZ^E_+$, and $\epsilon > 0$.}
    \ENSURE{$\bix \in \bbZ_+^E$ with $\biw^\top \bix \leq 1$.}
    \STATE{$\bix \leftarrow \bix_0$, $\biu \leftarrow \bic$, $d \leftarrow \max_{e \in E}\frac{f(\bichi_e)}{\biw(e)}$, and $\wmin \leftarrow \min_{e\in E}\biw(e)$.}
    \FOR{($\theta=d$; $\theta\geq \epsilon d\wmin$; $\theta \leftarrow \theta(1-\epsilon)$)}
      \FOR{all $e \in E$}
      \STATE{Find maximum $k \leq \biu(e) - \bix(e)$ with $f(k \bichi_e \mid \bix) \geq k\biw(e)\theta$ with binary search.}\label{line:knapsack-binary-search}
      \IF{$k > 0$}
        \STATE{\textbf{if} $\biw^\top\bix + k\biw(e) \leq 1$ \textbf{then} $\bix \leftarrow \bix + k \bichi_e$ \textbf{else} $\biu(e) \leftarrow \bix(e) + k - 1$.} \label{line:knapsack-trial}
        \COMMENT{``trial''}
      \ENDIF
      \ENDFOR
    \ENDFOR
    \RETURN{$\bix$.}
  \end{algorithmic}
\end{algorithm}


Let us fix an initial solution $\bix_0$ and analyze the behavior of Algorithm~\ref{alg:knapsack-greedy} on $\bix_0$.
Let us call an execution of Line~\ref{line:knapsack-trial} a \emph{trial}.
Let $e_i$ and $k_i$ be the value of $e$ and $k$ in the $i$-th trial.
We denote by $\bix_i$ the tentative solution $\bix$ following the $i$-th trial.
Assume that Algorithm~\ref{alg:knapsack-greedy} first has not updated the tentative solution $\bix$ in the $L$-th trial.
Equivalently, let $L$ be the minimum number such that $\bix_{L-1}=\bix_L$ and $\bix_{i-1} < \bix_i$ for $i=1,\dots, L-1$.
We consider only such a situation, because if this is not the case, then Algorithm~\ref{alg:knapsack-greedy} outputs a feasible solution $\bix$ satisfying
\[
    \OPT - f(\bix) \leq f(\bix \vee \bix^*) - f(\bix) \leq \sum_{e \in \{\bix^*\}\setminus\{\bix\}} f(\bichi_e \mid \bix) < \epsilon d \leq \epsilon \OPT,
\]
which means that $\bix$ gives $(1-\epsilon)$-approximation.
Thus we focus on the situation in which Algorithm~\ref{alg:knapsack-greedy} fails to increase the tentative solution.

\begin{lemma}\label{lem:assuption of e_L}
Without loss of generality, we may assume that $\bix_{L-1}(e_L) + k_L \leq \bix^*(e_L)$.
\end{lemma}
\begin{proof}
Suppose that $\bix_{L-1}(e_L) + k_L > \bix^*(e_L)$.
Let us consider a modified instance in which $\bic(e_L)$ is reduced to $\bix_{L-1}(e_L) + k_L - 1$.
The optimal value is unchanged by this modification because $\bix^*$ is still feasible and optimal.
Furthermore, Algorithm~\ref{alg:knapsack-greedy} returns the same solution.
Thus, it suffices to analyze the algorithm in the modified instance.
Repeating this argument completes the proof of this lemma.
\qed
\end{proof}

\begin{lemma}\label{lem:knapsack-update-gain}
  For $i = 1, \dots, L$, the average gain satisfies the following.
  \[
      \frac{f(k_i\bichi_{e_i} \mid \bix_{i-1})}{k\biw(e_i)} \geq (1 - \epsilon) \frac{f(\bichi_s \mid \bix_{i-1})}{\biw(s)} \quad (s \in \supp^+(\bix^* - \bix))
  \]
\end{lemma}
\begin{proof}
  The proof is similar to Lemma~\ref{lem:cardinality-update-gain}.
  For the sake of simplicity, let us fix $i$ and denote $\bix := \bix_{i-1}$, $e := e_i$, and $k := k_i$.
  We first have $f(k\bichi_{e} \mid \bix)  \geq k\biw(e)\theta $.
  Then, we show that $f(\bichi_s \mid \bix) \leq \frac{\biw(s)\theta}{1-\epsilon}$ for any $s \in \supp^+(\bix^* - \bix)$.
    This is trivial by DR-submodularity if $\theta = d$.
    Thus we assume that $\theta < d$, i.e., there is at least one threshold update.
    Let $s\in\{\bix^*\} \setminus \{\bix\}$, $k'$ be the increment in the $s$-th entry in the previous threshold (i.e., $\frac{\theta}{1 - \epsilon}$), and $\bix'$ be the variable $\bix$ at the time.
    Suppose that $f(\bichi_s \mid \bix) > \frac{\biw(s)\theta}{1-\epsilon}$.
Then $f((k'+1)\bichi_s \mid \bix') \geq f(\bichi_s \mid \bix) + f(k'\bichi_s \mid \bix') > \frac{\biw(s)\theta}{1-\epsilon} + \frac{k'\biw(s)\theta}{1-\epsilon} = \frac{(k'+1)\biw(s)\theta}{1-\epsilon}$, which contradicts the fact that $k'$ was the largest value with $f(k'\bichi_s \mid \bix')\geq \frac{k'\biw(s)\theta}{1-\epsilon}$.
  Eliminating $\theta$ from these inequalities completes the proof.
  \qed
  \end{proof}

\begin{lemma}\label{lem:greedy-last}
Let $\bix$ be the output of Algorithm~\ref{alg:knapsack-greedy} with an initial solution $\bix_0$. Then
\begin{align}
    f(\bix) \geq \Bigl(1-\frac{1}{e}-O(\epsilon)\Bigr)\OPT + \frac{f(\bix_0)}{e} - f(k_L\bichi_{e_L} \mid \bix_L).
\end{align}
\end{lemma}
\begin{proof}
By monotonicity and DR-submodularity, we have
\begin{align*}
    \OPT
    &\leq f(\bix_i\vee \bix^*) \\
    &\leq f(\bix_i) + \sum_{s\in\{\bix^*\}\setminus\{\bix_i\}}f(\bichi_s \mid \bix_{i-1})  \\
    &\leq f(\bix_i) +  \sum_{s\in\{\bix^*\}\setminus\{\bix_i\}}\frac{\biw(s)}{1-\epsilon} \frac{f(k_i\bichi_{e_i} \mid \bix_{i-1})}{k_i\biw(e_i)} \\
    &\leq f(\bix_i) +  \frac{1}{1-\epsilon} \frac{f(k_i\bichi_{e_i} \mid \bix_{i-1})}{k_i\biw(e_i)},  \tag{since $\sum_{s\in\{\bix^*\}\setminus\{\bix_i\}}\biw(s) \leq 1$}
\end{align*}
for $i = 1, \dots, L$.
Rearranging the terms, we have
\[
    f(k_i\bichi_{e_i} \mid \bix_{i-1}) \geq (1-\epsilon)\biw(e_i)k_i\left(\OPT -f(\bix_i)\right). \quad (i = 1, \dots, L)
\]
Then by induction, we can prove
\begin{align*}
    &\OPT - \sum_{i = 1}^L f(k_i\bichi_i \mid \bix_{i-1}) \\
    &\leq (\OPT - f(\bix_0)) \prod_{j=1}^L (1 - (1-\epsilon)\biw(e_j)k_j) \\
    &\leq (\OPT - f(\bix_0)) \exp \left( - (1-\epsilon)\sum_{j=1}^L \biw(e_j)k_j \right) \\
    &\leq (\OPT - f(\bix_0)) \exp \left( - 1 + \epsilon \right) \tag{since $\sum_{j=1}^L \biw(e_j) k_j > 1$} \\
    &\leq (\OPT - f(\bix_0))\Bigl(\frac{1}{e} + O(\epsilon)\Bigr).
\end{align*}
By monotonicity,
\begin{align*}
    f(\bix)
    &\geq f(\bix_L) \\
    &= f(\bix_0) + \sum_{i = 1}^L f(k_i\bichi_{e_i} \mid \bix_{i-1}) -  f(k_L\bichi_{e_L} \mid \bix_{L-1})\\
    &\geq \Bigl(1-\frac{1}{e}-O(\epsilon)\Bigr)\OPT + \frac{f(\bix_0)}{e} - f(k_L\bichi_{e_L} \mid \bix_{L-1}),
\end{align*}
which completes the proof. \qed
\end{proof}

\subsection{Partial Enumeration}\label{sec:knapsack-partial-enumeration}
\begin{algorithm}[t!]
    \caption{\textsf{PartialEnumeration}$(f, \bic, \biw, \epsilon)$}\label{alg:knapsack-enumeration}
  \begin{algorithmic}[1]
    \REQUIRE{$f:\bbZ_+^E \to \bbR_+$, $\bic \in \bbZ_+^E$, $\biw \in (0,1]^E$, and $\epsilon > 0$.}
    \ENSURE{A set $\mathcal{X}$ consisting of $\bix_0 \in \bbZ_+^E$ with $|\supp(\bix)|\leq 3$ and $\biw^\top \bix_0 \leq 1$.}
    \STATE{$\mathcal{X} \leftarrow \emptyset$.}
    \FOR{each ordered tuple $X$ consisting of at most three elements in $E$}
    \STATE{$\mathcal{Y}\leftarrow \{\bfzero\}$}
      \FOR{$i = 1, \dots, |X|$}
      \STATE{Let $e$ be the $i$-th element of $X$.}
      \STATE{$\mathcal{Y} \leftarrow \textsf{IncreaseSupport}(f,\biw,e,\mathcal{Y},\epsilon)$.}
      \ENDFOR
    \STATE{$\mathcal{X} \leftarrow \mathcal{X} \cup \{\bix \in \mathcal{Y} : \biw^\top\bix \leq 1 \}$.}
    \ENDFOR
    \RETURN{$\mathcal{X}$.}
  \end{algorithmic}
\end{algorithm}

\begin{algorithm}[t!]
  \caption{\textsf{IncreaseSupport}$(f,\bic, \biw,e,\mathcal{Y},\epsilon)$}
  \begin{algorithmic}[1]
  \REQUIRE{$f:\bbZ_+^E \to \bbR_+$, $\bic \in \bbZ_+^E$, $\biw \in (0,1]^E$, $e\in E$, $\mathcal{Y} \subseteq \bbZ_+^E$, and $\epsilon > 0$.}
  \ENSURE{A set $\mathcal{X}$.}
  \STATE{$\mathcal{X} \leftarrow \emptyset$.}
  \FOR{$\biy \in \mathcal{Y}$}
  \STATE{Find $k_{\min}$ with $0 \leq k_{\min} \leq \bic(e)$ such that $f(k_{\min}\bichi_e \mid \biy) > 0$ by binary search.}
  \STATE{\textbf{if} no such $k_{\min}$ exists \textbf{then} \textbf{continue}.}
  \FOR{ ($h = f(\bic(e) \bichi_e \mid \biy)$; $h \geq (1-\epsilon)f(k_{\min} \bichi_e \mid \biy) $; $h = (1-\epsilon)h$) }\label{line:for-binary-search-cardinality}
  \STATE{Find the smallest $k$ with $k_{\min} \leq k \leq \bic(e)$ such that $f(k \bichi_e) \geq h$ by binary search.}
  \STATE{Add $\biy + k\chi_e$ to $\mathcal{X}$.}
  \ENDFOR
  \ENDFOR
  \RETURN{$\mathcal{X}$.}
  \end{algorithmic}
\end{algorithm}

We now prove that the greedy procedure returns a $(1-1/e)$-approximate solution for some $\bix_0$ and that such $\bix_0$ can be found in polynomial time.
Here, we exploit partial enumeration technique.
The pseudocode description of our algorithm is shown in Algorithm~\ref{alg:knapsack-enumeration}.

\begin{lemma}\label{lem:partial-enumeration}
    There exists $\bix_0$ in the output of \textsf{PartialEnumeration} such that $f(k_L\bichi_{e_L} \mid \bix_{L-1}) \leq \frac{f(\bix_0)}{3(1-\epsilon)}$,
    where $\bix_{L-1}$, $k_L$, and $e_L$ are defined as above for Algorithm~\ref{alg:knapsack-greedy} with the initial solution $\bix_0$.
\end{lemma}
\begin{proof}
    In what follows, we assume that $n \geq 3$ for simplicity.\footnote{If $n < 3$, by a similar argument, one can show that there exists $\bix_0$ in the output of $\textsf{PartialEnumeration}$ that attains $(1-\epsilon)$-approximation.}
For a function $g:\bbZ_+^E \to \bbR_+$ and an element $e \in E$, we define
\begin{align*}
    H(g, e) := \{ (1-\epsilon)^s g(\bic(e)\bichi_e)  : s \in \bbZ_+, (1-\epsilon)^s g(\bic(e)\bichi_{e}) \geq g(k_{\min}\bichi_{e}) \},
\end{align*}
where $k_{\min}$ is the minimum $k \in \bbZ_+$ such that $g(k\bichi_{e}) > 0$.
Let us define $h(g, e)$ to be the unique element in $H(g,e)$ with $h(g, e) \leq g(\bix^*(e)\bichi_e) < \frac{h(g,e)}{1-\epsilon}$ and $l(g, e)$ to be the minimum integer such that $g(l(g,e)\bichi_e) \geq h(g,e)$.

Then we define $e_1^*, e_2^*, e_3^* \in E$ as follows:
\begin{align*}
    f_1 &:= f, \quad  e_1^* \in \argmax_{e\in E} f_1(\bix^*(e)\bichi_e), \quad l_1:= l(f_1, e_1^*) \\
    f_i &:= f(\cdot \mid \vee_{j=1}^{i-1} l_j\bichi_{e_j^*}), \quad  e_i^* \in \argmax_{e\in E} f_i(\bix^*(e)\bichi_e),\quad  l_i:= l(f_i, e_i^*) \quad (i = 2,3).
\end{align*}
Let us define $\bix_0 := \sum_{i=1}^3 l_i\bichi_{e_i^*}$.
Note that $\bix_0$ is an element of the output of \textsf{PartialEnumeration}.
By the definition, we have
\begin{align*}
    f(l_i\bichi_{e_i^*} \mid \vee_{j=1}^{i-1} l_j\bichi_{e_j^*})
    &\leq
    f(\bix^*(e_i^*)\bichi_{e_i^*} \mid \vee_{j=1}^{i-1} l_j\bichi_{e_j^*})
    \leq
    \frac{1}{1-\epsilon}f(l_i\bichi_{e_i^*} \mid \vee_{j=1}^{i-1} l_j\bichi_{e_j^*}).
\end{align*}
for $i = 1,2,3$.
We now show this $\bix_0$ satisfies the required condition.
We have
\begin{align*}
    f(k_L\bichi_{e_L} \mid \bix_{L-1})
    &\leq f(\bix^*(e_L)\bichi_{e_L} \vee \bix_{L-1}) - f(\bix_{L-1}) \tag{by Lemma~\ref{lem:assuption of e_L}} \\
    &\leq f(\bix^*(e_L)\bichi_{e_L} \vee \bfzero) - f(\bfzero) \tag{by weak diminishing return} \\
    &\leq f(\bix^*(e_1^*)\bichi_{e_1^*}) \tag{by the definition of $e_1^*$} \\
    &\leq \frac{1}{1-\epsilon}f(l_1\bichi_{e_1^*}).
\end{align*}
Similarly, applying weak diminishing return with $l_1\bichi_{e_1^*}$ instead of $\bfzero$, we obtain
\begin{align*}
    f(k_L\bichi_{e_L} \mid \bix_{L-1})
    &\leq f(\bix^*(e_2^*)\bichi_{e_L^*} \mid l_1\bichi_{e_1^*}) \leq \frac{1}{1-\epsilon}f(l_2\bichi_{e_2^*}\mid l_1\bichi_{e_1^*}).
\end{align*}
In the same way, we have
\begin{align*}
    f(k_L\bichi_{e_L} \mid \bix_{L-1})
    &\leq \frac{1}{1-\epsilon}f(l_3\bichi_{e_3^*}\mid l_1\bichi_{e_1^*} \vee l_2\bichi_{e_2^*}).
\end{align*}
Adding these inequalities, we obtain $3 f(k_L\bichi_{e_L} \mid \bix_{L-1}) \geq f(\bix_0)/(1-\epsilon)$.
\qed
\end{proof}

\subsection{Final Algorithm}\label{sec:knapsack-final}
Our final algorithm for maximizing a montone DR-submodular function subject to a knapsack constraint is shown in Algorithm~\ref{alg:knapsack-DR}.

\begin{algorithm}[]
    \caption{Knapsack Constraint/DR-Submodular}\label{alg:knapsack-DR}
    \begin{algorithmic}[1]
        \REQUIRE{$f:\bbZ_+^E \to \bbR_+$, $\bic \in \bbZ_+^E$, $\biw \in (0,1]^E$, and $\epsilon > 0$.}
        \ENSURE{$\bix \in \bbZ_+^E$.}
        \STATE{$\mathcal{X} \leftarrow \textsf{PartialEnumeration}(f,\bic, \biw,\epsilon)$, $\mathcal{G} \leftarrow \emptyset$.}
        \FOR{each $\bix_0 \in \mathcal{X}$}
            \STATE{$\biy \leftarrow \textsf{GreedyKnapsack}(f, \bic, \biw,\bix_0,\epsilon)$}
            \STATE{Add $\biy$ to $\mathcal{G}$.}
        \ENDFOR
        \STATE{$\bix \leftarrow \argmax_{\biy \in \mathcal{G}} f(\biy)$.}
        \RETURN{$\bix$}
    \end{algorithmic}
\end{algorithm}

\begin{theorem}
  Algorithm~\ref{alg:knapsack-DR} finds a $(1-1/e-O(\epsilon))$-approximate solution in
  \[
      \compknap
  \] time,
  where $\tau = \frac{ \max_{e \in E} f(\bic(e)\bichi_e) }{ \min\{ f(\bichi_e \mid \bix) : e \in E, \bfzero \leq \bix \leq \bic, f(\bichi_e \mid \bix) > 0 \}}$, $\wmin = \min_{e\in E} \biw(e)$, and $0 < \epsilon < 1 - e/3$.
\end{theorem}
\begin{proof}
    Let $\bix_0$ be the element in the output of \textsf{PartialEnumeration} described in Lemma~\ref{lem:partial-enumeration} and $\bix_{L-1}$ be the corresponding variable for \textsf{GreedyKnapsack} with the initial solution $\bix_0$.
  By Lemmas~\ref{lem:greedy-last} and \ref{lem:partial-enumeration}, 
  \begin{align*}
    f(\bix_{L-1}) & \geq \Bigl(1-\frac{1}{e}-O(\epsilon)\Bigr)\OPT + \Bigl(\frac{1}{e} - \frac{1}{3(1-\epsilon)}\Bigr)f(\biy_0) \\
    & \geq \Bigl(1-\frac{1}{e}-O(\epsilon)\Bigr)\OPT.
  \end{align*}

  For the running time, \textsf{PartialEnumeration} finds
  $O(\frac{n^3}{\epsilon^3} \log^3 \tau)$
  initial solutions in
  $O(\frac{n^3}{\epsilon^3}\log^3\|\bic\|_\infty \log^3\tau)$ time.
  For each initial solution, \textsf{GreedyKnapsack} takes $O(\frac{n}{\epsilon}\log\|\bic\|_\infty \log \frac{1}{\epsilon\wmin})$ time.
  Thus the total running time of Algorithm~\ref{alg:knapsack-DR} is as claimed.
  \qed
\end{proof}

\bibliographystyle{spmpsci}      
\bibliography{sfm}

\end{document}